\documentclass[lettersize,journal]{IEEEtran}
\usepackage[colorlinks, linkcolor=red, citecolor=blue]{hyperref}
\usepackage{amssymb,amsthm,amsmath,latexsym,pslatex,cite}
\usepackage{lineno,hyperref,mathtools}
\usepackage[ruled,linesnumbered]{algorithm2e}
\usepackage{pdflscape}
\usepackage{microtype}
\usepackage{setspace}
\usepackage{amsmath}
\usepackage{amssymb}
\usepackage{bm}
\usepackage{amssymb}
\usepackage{amsthm}
\usepackage{multirow,booktabs}
\usepackage{cases}
\usepackage{tabularx}
\usepackage{adjustbox}
\usepackage[figuresright]{rotating}
\usepackage{longtable}
\usepackage{booktabs}
\usepackage{multirow}
\usepackage{color}
\usepackage{cite}
\usepackage{multirow,booktabs}
\usepackage{cases}
\usepackage{tabularx}
\usepackage{adjustbox}
\usepackage[figuresright]{rotating}
\usepackage{longtable}
\usepackage{booktabs}
\usepackage{multirow}
\usepackage{color}
\usepackage{lineno,hyperref,mathtools}
\usepackage{soul}
\usepackage{algpseudocode}

\allowdisplaybreaks

\newtheorem{theorem}{Theorem}

\newtheorem{problem}{Problem}
\newtheorem{lemma}[theorem]{Lemma}

\theoremstyle{definition} 
\newtheorem{remark}{Remark}
\newtheorem{example}{Example}
\newtheorem{definition}{Definition}

\theoremstyle{remark}

\newcommand{\F}{\mathbb{F}}

\newcommand{\C}{{\mathcal{C}}}
\newcommand{\D}{{\mathcal{D}}}
\newcommand{\SSS}{{\mathcal{S}}}
\newcommand{\diag}{{\mathrm{diag}}}
\newcommand{\GRS}{{\mathrm{GRS}}}

\newcommand{\supp}{{\mathrm{supp}}}

\newcommand{\uuu}{{{\mathbf{u}}}}
\newcommand{\vvv}{{{\mathbf{v}}}}

\newcommand{\wt}{{{\rm{wt}}}}
\newcommand{\w}{{\omega}}

\newcommand{\ignore}[1]{}

\makeatletter
\newcommand{\rmnum}[1]{\romannumeral #1}
\newcommand{\Rmnum}[1]{\expandafter\@slowromancap\romannumeral #1@}
\makeatother

\begin{document}

\title{Unique Decoding of Extended Subcodes of GRS Codes Using Error-Correcting Pairs}
\author{Yang Li, Zhenliang Lu, San Ling, Shixin Zhu, Kwok Yan Lam 
\thanks{
    This research is supported by the Nanyang Technological University Research under Grant No. 04INS000047C230GRT01, 
    the National Natural Science Foundation of China under Grant Nos. U21A20428 and 12171134, and   
    the National Research Foundation, Singapore and Infocomm Media Development Authority 
    under its Trust Tech Funding Initiative. 
    Any opinions, findings and conclusions or recommendations expressed in this material are those of the authors 
    and do not reflect the views of National Research Foundation, Singapore and Infocomm Media Development Authority. 
    {\em (Corresponding author: Zhenliang Lu)}}
\thanks{Yang Li is with the School of Physical and Mathematical Sciences, 
Nanyang Technological University, 21 Nanyang Link, Singapore 637371, Singapore
(email: yanglimath@163.com).} 
\thanks{Zhenliang Lu is with the Department of Computer Science, City University of Hong Kong, 
Hong Kong, China (email: zhenliang.lu@cityu.edu.hk).}
\thanks{San Ling is with the School of Physical and Mathematical Sciences, 
Nanyang Technological University, 21 Nanyang Link, Singapore 637371, Singapore
(email: lingsan@ntu.edu.sg). 
He is also with VinUniversity, Vinhomes Ocean Park, Gia Lam, Hanoi 100000, Vietnam (email: ling.s@vinuni.edu.vn).} 
\thanks{Shixin Zhu is with the School of Mathematics, Hefei University of Technology, Hefei 230601, China 
(email: zhushixin@hfut.edu.cn).}
\thanks{Kwok Yan Lam is with the Digital Trust Centre,
Nanyang Technological University, 50 Nanyang Drive, Singapore 639798, Singapore 
(email: kwokyan.lam@ntu.edu.sg).}
}



\maketitle
\begin{abstract}
  Extended subcodes of generalized Reed-Solomon (ESGRS) codes are a class of linear codes where each code is either  
  a non-GRS maximum distance separable (MDS) code or a near MDS (NMDS) code.  
  They have important applications in distributed storage systems and cryptography. 
  While many algebraic properties and explicit constructions of ESGRS codes 
  have been well studied in the literature, their decoding has been unexplored. 
  In this paper, we focus on their unique decoding problems in terms of $\ell$-error-correcting pairs 
  ($\ell$-ECPs). 
  We determine the existence and specific forms of their $\ell$-ECPs by utilizing GRS codes and their variants.   
  We further present an explicit decoding algorithm for ESGRS codes based on these $\ell$-ECPs, 
  which can uniquely correct up to $\ell$ errors in polynomial time, with $\ell$ about half of the minimum distance. 
  Compared with other potential decoding approaches, our algorithm is more efficient in certain scenarios. 
  Some concrete examples are also given to support these results.
\end{abstract}


\section{Introduction}\label{sec.introduction}

Maximum distance separable (MDS) and near MDS codes are two special types of codes in coding theory, 
which have recently garnered significant attention due to their elegant algebraic structures 
and wide-ranging applications in distributed storage systems \cite{CHL2011}, random error channels \cite{ST2013}, 
informed source and index coding problems \cite{TR2018}, and secret sharing schemes \cite{SV2018, ZWXLQY2009}.
Formally, an $[n,k,d]_q$ {\em linear code} $\C$ is a $k$-dimensional linear subspace of $\F_q^n$ 
with minimum distance $d$, and it satisfies the {\em Singleton bound}: $d\leq n-k+1$ \cite{HP2003},   
where $\F_q$ is the finite field of size $q$, $\F_q^n$ is the $n$-dimensional vector space, and $q=p^m$ is a prime power. 
A code $\C$ achieving this Singleton bound, $i.e.$, $d = n - k + 1$, is called a {\em maximum distance separable (MDS)} code. 
If $d = n - k$, then $\C$ is referred to as an {\em almost MDS (AMDS)} code. 
Moreover, a code $\C$ is called a {\em near MDS (NMDS)} code if both $\C$ and $\C^{\perp}$ are AMDS codes, where $\C^{\perp}$ is the dual code of $\C$ with respect to a certain inner product.

\subsection{Linear codes where each code is either non-GRS MDS or NMDS and decoding problems}

Generalized Reed-Solomon (GRS) codes are a special class of MDS codes that 
have garnered a lot of attention due to their highly efficient encoding and decoding algorithms. 
Moreover, $[n,k,n-k+1]_q$ GRS codes satisfy $n \leq q$, and together with the MDS conjecture $n \leq q+1$ 
in general except for two special cases \cite{HP2003}, 
they realize nearly all possible parameter sets of MDS codes, except for the case $n > q$.
However, it is still very interesting and vital to construct 
MDS codes that are not monomially equivalent to GRS codes, referred to as {\em non-GRS MDS codes}. 
These non-GRS MDS codes are important in classifying MDS codes, and are known to exhibit significant potential in resisting cryptographic attacks such as the Sidelnikov–Shestakov and Wieschebrink attacks, 
unlike GRS codes \cite{LR2020, BBPR2018}.
Consequently, it is of particular interest to construct new families of linear codes in which each code is either a non-GRS MDS code or an NMDS code. 

There has been a substantial body of work 
\cite{LZS2025,WDC2023,HL2024,HYNL2021,HF2023,HYN2023,ZL2024,HZ2024,SYLH2022,SD2023,SDC2024,LSZ2025,C2024,ZZ2025,JMXZ2024,LCCN2025,WHLD2024} 
focusing on the construction of new families of codes that include both MDS and NMDS codes. 
In particular, the (+)-twisted generalized Reed-Solomon ((+)-TGRS) codes in \cite{HYNL2021}, 
the ($\ast$)-TGRS codes in \cite{HYN2023}, the (+)-extended TGRS ((+)-ETGRS) codes in \cite{ZL2024}, 
and the extended subcodes of GRS (ESGRS) codes in \cite{LSZ2025} coincide with examples of linear code families where each code is either a non-GRS MDS code or an NMDS code.

As Huffman and Pless emphasized in \cite{HP2003}, finding efficient (fast) decoding algorithms 
is a major area of research in coding theory because of their practical applications. 
There are many different decoding approaches, such as unique decoding and non-unique decoding, which can be classified according to their decoding outputs.
Compared with non-unique decoding, unique decoding guarantees both correctness and unambiguity of the decoding output within its designed decoding radius \cite{G2001}. 
Specifically, when the number of errors does not exceed half of the minimum distance, 
unique decoding deterministically recovers the transmitted codeword and excludes any decoding ambiguity. 
This property is particularly crucial in cryptographic and security-sensitive applications, 
where multiple candidate outputs are unacceptable \cite{M1978, S2017}. 
Therefore, in this paper, we focus on the unique decoding of such families of linear codes 
where each code is either non-GRS MDS or NMDS. 
For convenience, unless stated otherwise, all subsequent references to ``decoding" refer to {\em unique decoding}.

\subsection{Our motivations for unique decoding ESGRS codes based on ECPs}

{\it Of the codes in the previous subsection, except for the ESGRS codes, 
unique decoding algorithms have been studied for the other three families of linear codes.} 
Let $q$ be the alphabet size and $n$ be the code length with $n\leq q$. 
Specifically, Beelen $et~al.$ \cite{BPN2017} initially proposed a brute-force decoding method 
based on the established decoding algorithm of GRS codes when introducing TGRS codes with time complexity $O(q^t n^2)$ \cite{SYJL2024}, 
where $t$ is the number of twists. A simpler version was later presented in \cite{BPR2022}. 
    Although these decoding approaches are theoretically applicable for decoding (+)-TGRS and ($\ast$)-TGRS codes, 
    they are not practical in real-world scenarios. 
    {As stated in \cite{BPN2017} and \cite{BPR2022}, the former suffers from unexpected decoding complexity, 
    while the latter lacks a rigorous proof ensuring that decoding works 
    for any error vector up to the maximal decoding radius.}

    Additionally, Jia $et~al.$ \cite{JYS2025}, Sui $et~al.$ \cite{SY2023}, and Sun $et~al.$ \cite{SYJL2024} 
    applied the Berlekamp–Massey algorithm for GRS codes and the extended Euclidean algorithm to propose efficient 
    decoding algorithms for the dual codes of MDS (+)-TGRS and MDS ($\ast$)-TGRS codes. 
    Wang $et~al.$ \cite{WLL2025} further extended it to the NMDS cases.    
    According to \cite{SYJL2024}, this method has time complexity $O(qn)$ in general. 
    {It is also worth noting from \cite{HYNL2021} and \cite{HYN2023} that the dual codes of 
    (+)-TGRS and ($\ast$)-TGRS codes may not be (+)-TGRS or ($\ast$)-TGRS codes again.} 
Consequently, this method cannot be straightforwardly applied to arbitrary 
(+)-TGRS and ($\ast$)-TGRS codes.

    Alternatively, He $et~al.$ \cite{HL2023DM} explored the forms of $\ell$-error-correcting pairs ($\ell$-ECPs) 
    for general (+)-TGRS and ($\ast$)-TGRS codes.  
    Li $et~al.$ theoretically studied the existence of $\ell$-ECPs for MDS (+)-ETGRS codes in \cite{LZS2025}   
    and Li $et~al.$ provided explicit constructions of $\ell$-ECPs for NMDS (+)-ETGRS codes in \cite{LELL2025}, 
    where $\ell$ is approximately half of the minimum distance.  
    It can be observed that the resulting decoding algorithms can uniquely correct up to $\ell$ errors in 
    polynomial time on $n$, independent of the alphabet size $q$, 
    thereby making them much faster than other known approaches in certain scenarios.

From the above discussions, we can immediately summarize the following three important facts: 
\begin{itemize}
    \item{\it Fact 1:} The current known families of linear codes where each code is either a non-GRS MDS code or an NMDS code  
    mainly stem from TGRS codes and their extended codes.

    \item{\it Fact 2:}  For the aforementioned Fact 1 codes, decoding algorithms based on ECPs appear 
    to be more practical and versatile than other known approaches.

    \item{\it Fact 3:} Although the algebraic properties of ESGRS codes have been extensively studied in \cite{LSZ2025}, 
    their decoding problems remain unexplored. In particular, it was also proved in \cite{LSZ2025} that an ESGRS code  
is generally not monomially equivalent to a TGRS code or its extended code, due to different weight distributions. 
\end{itemize}

The concept of ECP was introduced 
by K\"otter \cite{K1992} and Pellikaan \cite{P1992} in 1992 (see Definition \ref{def.ECP}).  
It has also been further demonstrated that ECPs exist for a wide range of linear codes, 
including cyclic codes, GRS codes, TGRS codes with single or double twists, 
lengthened GRS codes, Goppa codes, alternant codes, and algebraic geometry codes 
\cite{DK1994,K1992,P1992,P1996,HL2023DM,HL2024,XL2025,LZS2025}.
In particular, if a linear code possesses an $\ell$-ECP, this framework enables a universal algebraic decoding  
capable of uniquely correcting up to $\ell$ errors in polynomial time, where the specific time complexity depends on 
the related extension degree of finite fields \cite{P1996,DK1994}. 
Furthermore, Pellikaan $et~al.$ explored the application of ECPs in the 
McEliece public-key cryptosystem for quantum-resistant cryptography \cite{PM2017}. 
In consideration of the above facts and applications, an important and intriguing problem naturally arises:

\begin{problem}\label{prob2}
How can we decode ESGRS codes by leveraging their ECPs? 
More specifically, is it possible to develop an efficient decoding algorithm for ESGRS codes based on their ECPs?
\end{problem}

\subsection{Our contributions}


We answer the aforementioned {\bf Problem \ref{prob2}} affirmatively,
and our main contributions can be summarized as follows.  

\begin{itemize}
    \item [\rm 1)] 
    We use $\C_k(\SSS,{\bf v},\infty)$ to denote an $[n+1,k]_q$ ESGRS code 
    (see Definition \ref{def.codes} for more details). 
    Let $\ell=\left \lfloor \frac{d(\C_k(\SSS,{\bf v},\infty))-1}{2} \right \rfloor$, 
    where  $d(\C_k(\SSS,{\bf v},\infty))$ denotes the minimum distance of 
    the ESGRS code $\C_k(\SSS,{\bf v},\infty)$. 
    Utilizing structural properties, we identify 
    ESGRS codes as subcodes of certain GRS codes and EGRS codes 
    in Lemma \ref{lem.nested_relationship}. 
    Building on this, 
    we prove in Theorem \ref{th.ECP_of_C_k}.1) that $\C_k(\SSS,{\bf v},\infty)$ has 
    only $(\ell-1)$-ECPs if it is an $[n+1,k,n-k+2]_q$ MDS code and $2\nmid (n-k)$.  
    In this case, the ECP can be determined by two GRS codes. 

    \item [\rm 2)]
    We also show that $\C_k(\SSS,{\bf v},\infty)$ has $\ell$-ECPs 
    in other cases by providing explicit constructions of these ECPs in Theorems \ref{th.ECP_of_C_k}.2) 
    and \ref{th.ECP_of_C_k222}, based on variants of GRS codes. 
    More details on the exact values of $\ell$ and explicit forms of ECPs 
    are summarized in Table \ref{tab.ECP}.  
    
    \item [\rm 3)] We provide an explicit decoding algorithm for $\C_k(\SSS,{\bf v},\infty)$ with time complexity $O(n^3)$ 
    based on their ECPs in Theorem \ref{th.decoding ECP} and Algorithm \ref{alg.1}. 
    We further compare our decoding algorithm with other potential approaches and 
    identify some explicit advantages in Subsection \ref{subsec.complexity}. 
    We also present two specific decoding examples shown in Examples \ref{exam.1} and \ref{exam.2} to illustrate our results. 
\end{itemize}

\medskip
This paper is organized as follows. 
After the introduction, Section \ref{sec2.preliminaries} reviews fundamental notation 
and results on GRS codes, ESGRS codes, and ECPs. 
Section \ref{sec.3} establishes the existence and explicit forms of 
$\ell$-ECPs for ESGRS codes. 
Section \ref{sec.4} presents an explicit decoding algorithm based on their $\ell$-ECPs and further discusses its complexity. 
Finally, Section \ref{sec.concluding remarks} concludes this paper.

\section{Preliminaries}\label{sec2.preliminaries}

In this section, we recall some basic notation and results on 
GRS codes, ESGRS codes, and ECPs.

\subsection{Basic notation}
From now on, we fix the following notation, unless stated otherwise. 
\begin{itemize}
    \item $\F_q$ is the {\em finite field} of size $q$ and $\F_q^n$ is the {\em $n$-dimensional vector space}, where $q=p^m$ is a prime power. 
    \item For an $[n,k,d]_q$ linear code $\C$, we use $\dim(\C)$ and $d(\C)$ to denote its {\em dimension} and {\em minimum distance}, respectively. 
    \item For distinct elements $a_1,\ldots,a_n\in \F_q$, we let $\SSS=\{a_1,\ldots,a_n\}$ and 
    let ${\bf u}=(u_1,\ldots,u_n)\in (\F_q^*)^n$ with 
    $$u_i=\prod_{1\leq j\leq n,~j\neq i}(a_i-a_j)^{-1}~{\rm for}~ 1\leq i\leq n.$$
    

    \item  ${\bf 0}$ (resp. ${\bf 1}$) is an appropriate {\em row} or {\em column vector} of all zeros (resp. ones), 
    $\infty_{k+1}=(0,\ldots,0,1)^T$ is a column vector of length $k+1$, 
    and $\mathbf{v}=(v_1,\ldots,v_n)\in (\F_q^*)^n$ is a generic representation of its element.

    \item For any vector ${\bf x}=(x_1,\ldots,x_n)\in \F_q^n$, 
    $\supp(x)=\{i:~ x_i\neq 0\}$ and $z(x)=\{i:~ x_i=0\}$.
    
    \item For any vectors ${\bf x}=(x_1,\ldots,x_n)\in \F_q^n$ and ${\bf y}=(y_1,\ldots,y_n)\in \F_q^n$, 
    their {\em ordinary inner product} and {\em Schur product} are respectively defined by  
    $$\langle{\bf x}, {\bf y}\rangle=\sum_{i=1}^{n}x_iy_i~{\rm and}~{\bf x} \star {\bf y}= (x_1y_1,\ldots,x_ny_n).$$  
    
    \item For any two linear codes $\C$ and $\D$ with the same length $n$, the {\em dual code} of $\C$ is given by 
    $$\C^{\perp}=\{\mathbf{c}^{\perp}\in \F_q^n:~ \langle\mathbf{c}, \mathbf{c}^{\perp}\rangle=0,~\forall~\mathbf{c}\in \C\}$$
    and the {\em Schur product code} of $\C$ and $\D$ is defined by 
    \begin{align*}
        \mathcal{C}\star \mathcal{D}=\{  (c_1d_1, & \ldots,c_nd_n): \\
         & (c_1,\ldots,c_n)\in \mathcal{C},~(d_1,\ldots,d_n)\in \mathcal{D}\}.
    \end{align*}


    \item $\C A=\{{\bf c}A:~{\bf c}\in \C\}$, where $\C$ is an $[n,k]_q$ linear code and $A$ is an $n\times n$ matrix over $\F_q$.

    \item $\diag(a_1,\ldots,a_n)$ denotes an $n\times n$ {\em diagonal matrix} with diagonal elements $a_1,\ldots,a_n$ over $\F_q$. 
    
    \item $\F_{q}[x]_k=\{f(x)=\sum_{i=0}^{k-1}f_ix^i:~f_i\in \F_q,~\forall~ i\in \{0,\ldots,k-1\}\}$.
    
    \item $\mathcal{V}_k=\{f(x)=\sum_{i=0}^{k-2}f_ix^i+f_kx^k:~f_i\in \F_q,~\forall~ i\in \{0,\ldots,k-2,k\}\}$. 
\end{itemize}

\subsection{GRS codes and ESGRS codes}

As mentioned before, an important class of MDS codes is the so-called generalized Reed-Solomon codes \cite{HP2003}.  
With the above notation, we have the following definition. 

\begin{definition}{\rm (\!\! \cite{HP2003})}
Let $k$ and $n$ be two positive integers satisfying $1\leq k\leq n\leq q$.  
A {\em generalized Reed-Solomon (GRS) code}, denoted by $\GRS_k(\SSS,\mathbf{v})$, is 
an $[n,k,n-k+1]_{q}$ linear code defined by 
$$
\GRS_k(\SSS,\mathbf{v}) =  \{(v_1f(a_1),\ldots,v_nf(a_n)):~ f(x)\in {\mathbb{F}_{q}[x]_{k}}\}.
$$
By adding a coordinate to each codeword of a GRS code, 
we immediately get an $[n+1,k,n-k+2]_q$ MDS code, 
namely, an {\em extended GRS (EGRS) code} 
$$
\GRS_k(\SSS,\mathbf{v},\infty) =  \{(v_1f(a_1),\ldots,v_nf(a_n),f_{k-1}):~ f(x)\in {\mathbb{F}_{q}[x]_{k}}\},
$$
where $f_{k-1}$ is the coefficient of $x^{k-1}$ in $f(x)$.     
\end{definition}

In general, we call $\SSS$ the {\em evaluation-point sequence} of 
$\GRS_k(\SSS,\mathbf{v})$ and $\GRS_k(\SSS,\mathbf{v},\infty)$.
From \cite{HP2003}, we can write generator matrices of $\GRS_k(\SSS,\mathbf{v})$ 
and $\GRS_k(\SSS,\mathbf{v},\infty)$ as 
\begin{align}\label{eq.GRS.generator matrix}
    G_{\GRS_k(\SSS,\mathbf{v})}=\begin{pmatrix}    
        v_1 & \ldots & v_{n} \\ 
        v_1a_1 &   \ldots & v_{n}a_{n} \\
        \vdots &   \ddots & \vdots  \\
        v_1a_1^{k-1} &  \ldots & v_{n}a_{n}^{k-1}  \\
    \end{pmatrix}
    \end{align}
and
\begin{align}\label{eq.EGRS.generator matrix}
    G_{\GRS_k(\SSS,\mathbf{v},\infty)}=\begin{pmatrix}    
        v_1 &   \ldots & v_{n} & 0 \\ 
        v_1a_1 &   \ldots & v_{n}a_{n} & 0 \\
        \vdots &   \ddots & \vdots & \vdots \\
        v_1a_1^{k-1} &  \ldots & v_{n}a_{n}^{k-1} & 1 \\
    \end{pmatrix},
\end{align}
respectively.

\begin{lemma}{\rm (\!\! \cite[Lemma 5]{FF2018} and \cite[Theorem 5.3.3]{HP2003})}\label{lem.GRS duality}
    The following statements hold. 
    \begin{enumerate}
        \item [\rm 1)] $\GRS_k(\SSS,\mathbf{v})^{\perp}=\GRS_{n-k}(\SSS,\mathbf{v}^{-1}\star \uuu)$.

        \item [\rm 2)] $\GRS_k(\SSS,\mathbf{v},\infty)^{\perp}=
        \GRS_{n-k+1}(\SSS,\mathbf{v}^{-1}\star \uuu,\infty)\cdot A$ with $A=\diag(\underbrace{1,\ldots,1}_n,-1)$.
    \end{enumerate}
    
\end{lemma}

\begin{definition}{\rm (\!\! \cite[Definition 1]{LSZ2025})}\label{def.codes}
    Let $n$ and $k$ be two positive integers satisfying $3\leq k\leq n-2\leq q-2$.
    An {\em extended subcode of GRS (ESGRS) code}, denoted by $\C_k(\SSS,{\bf v},\infty)$
        is an $[n+1,k]_q$ linear code 
        \begin{align*}
            \C_k(\SSS,{\bf v},\infty)=
            \{(v_1f(a_1),\ldots,v_nf(a_n),f_{k}): ~f(x)\in \mathcal{V}_k\},  
        \end{align*}
    where 
    $f_k$ is the coefficient of $x^k$ in $f(x)$. 
\end{definition}

The following lemmas recall some basic properties of ESGRS codes obtained in \cite{LSZ2025}.

\begin{lemma}{\rm (\!\! \cite[Theorem \Rmnum{3}.2]{LSZ2025})}\label{lem.parity check matrix}    
    The ESGRS code $\C_k(\SSS,{\bf v},\infty)$ has 
    a generator matrix and a parity-check matrix of the form 
        \begin{align}\label{eq.ESGRS_generator_matrix}
            G_{\C_k(\SSS,{\bf v},\infty)}=\begin{pmatrix}    
        v_1 &  \ldots & v_{n} & 0 \\ 
        v_1a_1 &   \ldots & v_{n}a_{n} & 0 \\
        \vdots &   \ddots & \vdots & \vdots \\
        v_1a_1^{k-2} &   \ldots & v_{n}a_{n}^{k-2} & 0 \\
        v_1a_1^k &  \ldots & v_{n}a_{n}^k & 1 \\
            \end{pmatrix} 
    \end{align}
and
 \begin{align}\label{eq.ESGRS_parity_matrix}
            H_{\C_k(\SSS,{\bf v},\infty)}=\begin{pmatrix}    
                \frac{u_1}{v_1} &   \ldots & \frac{u_n}{v_n} & 0 \\ 
                \frac{u_1}{v_1}a_1    & \ldots & \frac{u_n}{v_n}a_{n} & 0 \\
        \vdots &   \ddots & \vdots & \vdots \\
        \frac{u_1}{v_1}a_1^{n-k-2}  &  \ldots & \frac{u_n}{v_n}a_{n}^{n-k-2} & 0 \\
        \frac{u_1}{v_1}a_1^{n-k-1}  &  \ldots & \frac{u_n}{v_n}a_{n}^{n-k-1} & -1 \\
        \frac{u_1}{v_1}a_1^{n-k}  &    \ldots &  \frac{u_n}{v_n}a_{n}^{n-k} & -\sum_{i=1}^na_i
            \end{pmatrix},
        \end{align}
        respectively. 
\end{lemma}

For a subset $\SSS\subseteq \F_q$, we say that {\em $\SSS$ contains a $k$-zero-sum subset} if 
there exists a nonempty subset $\{a_{i_1},\ldots,a_{i_k} \} \subseteq \SSS$ such that 
$a_{i_1}+\ldots+a_{i_k}=0$. Conversely, we say that {$\SSS$ is {\em $k$-zero-sum free}}.

\begin{lemma}{\rm (\!\! \cite[Theorem \Rmnum{3}.4]{LSZ2025})}\label{lem.must be MDS or NMDS} 
    Any ESGRS code $\C_k(\SSS,{\bf v},\infty)$ is either MDS or NMDS. 
    Furthermore, the following statements hold. 
    \begin{enumerate}
        \item [\rm 1)] $\C_k(\SSS,{\bf v},\infty)$ is an $[n+1,k,n-k+2]_q$ MDS code if and only if $\SSS$ is $k$-zero-sum free. 
        \item [\rm 2)] $\C_k(\SSS,{\bf v},\infty)$ is an $[n+1,k,n-k+1]_q$ NMDS code if and only if $\SSS$ contains a $k$-zero-sum subset. 
    \end{enumerate}
\end{lemma}


\begin{lemma}{\rm (\!\! \cite[Theorem \Rmnum{3}.7]{LSZ2025})}\label{lem.non-GRS}
    Any ESGRS code $\C_k(\SSS,{\bf v},\infty)$ is non-GRS. 
\end{lemma}

\subsection{ECPs of linear codes}
This subsection reviews some basic definitions of ECPs . 
For any linear code $\C$, a pair of linear codes $(\mathcal{A},\mathcal{B})$ is 
an $\ell$-error-correcting pair ($\ell$-ECP) of $\C$ that can yield a decoding algorithm 
capable of uniquely correcting up to $\ell$ errors in polynomial time. 
From \cite[Corollary 2.15]{P1992}, we have $\ell\leq \lfloor \frac{d(\C)-1}{2} \rfloor$. 
The following is an alternative definition of the $\ell$-ECP.

\begin{definition}{\rm (\!\! \cite{DK1994,P1992})}\label{def.ECP}
    Let $\C$ be an $[n,k]_q$ linear code. 
    Let $\mathcal{A}$ and $\mathcal{B}$ be two linear codes with the same length $n$ over $\F_{q^t}$ for a given $t\in \mathbb{N}$.  
    We call $(\mathcal{A},\mathcal{B})$ an {\em $\ell$-error-correcting pair ($\ell$-ECP) of $\C$} over $\F_{q^t}$  
    if all the following conditions hold:
\begin{enumerate}
\item $\mathcal{A} \star \mathcal{B}\subseteq \C^{\perp}$, {with $\C^{\perp}$ seen as a linear code over $\F_{q^t}$}.
\item $d(\mathcal{B}^{\perp})> \ell$.
\item {$\dim_{\F_{q^t}}(\mathcal{A})> \ell$}.
\item $d(\mathcal{A})+d(\C)> n$.
\end{enumerate}
\end{definition}

The following three lemmas are useful.

\begin{lemma}{\rm (\!\! \cite[Remark 2.1]{HL2024})}\label{lem.subcode_ECP}
    Let $\C_i$ be an $[n,k_i]_q$ linear code for $i=1,2$ with $\C_1\subseteq \C_2$ 
    and $(\mathcal{A}, \mathcal{B})$ be an $\ell$-ECP of $\C_2$. 
    Then the following statements hold. 
    \begin{itemize}
        \item [\rm 1)] For any vector ${\bf x}\in (\F_q^*)^n$, 
        $({\bf x}^{-1}\star \mathcal{A}, {\bf x}\star \mathcal{B})$ is an $\ell$-ECP of $\C_2$. 

        \item [\rm 2)] $(\mathcal{A}, \mathcal{B})$ is an $\ell$-ECP of $\C_1$. 
    \end{itemize} 
\end{lemma}

\begin{lemma}{\rm (\!\! \cite[Example 4.2]{lem_ECP} and \cite[Theorem 6.2]{lem_ECP})}\label{lem.GRS_ECP}
    Let $\C$ be an $[n,n-2\ell,2\ell+1]_q$ MDS code, where $\ell$ is a positive integer with $1\leq \ell\leq \frac{n}{2}-1$. 
    Then $\C$ has an $\ell$-ECP $(\mathcal{A},\mathcal{B})$ over $\F_q$ if and only if $\C$ is a GRS code.  
    Moreover, both $\mathcal{A}$ and $\mathcal{B}$ are GRS codes with the same evaluation-point sequence as that of $\C$, 
    and $\mathcal{B}=(\mathcal{A}\star \mathcal{C})^{\perp}$. 
\end{lemma}

\begin{lemma}{\rm (\!\! \cite[Proposition 2.5]{P1996})}\label{lem.GRS_ECP_paraA}
    Let $\C$ be an $[n,n-2\ell,2\ell+1]_q$ MDS code and let $(\mathcal{A},\mathcal{B})$ be an $\ell$-ECP of $\C$. 
    Then $\mathcal{A}$ is an $[n,\ell+1,n-\ell]_q$ MDS code.
\end{lemma}


\section{The existence and specific forms of ECPs of ESGRS codes}\label{sec.3}
From now on, we explore the answers to {\bf Problem \ref{prob2}}. 
In this section, we study the existence and specific forms of ECPs of 
ESGRS codes in consideration of the MDS and NMDS cases.

\subsection{The case of MDS ESGRS codes}

First of all, we note that the following nested relationships 
between ESGRS codes and (extended) GRS codes.

\begin{lemma}\label{lem.nested_relationship}
    For any $3\leq k\leq k'-1\leq n-2\leq q-3$ and $\gamma \in \F_q\setminus \SSS$, 
    we have  
    $$
    \C_{k}(\SSS,\vvv,\infty)\subseteq \GRS_{k+1}(\SSS,\vvv,\infty)\subseteq \GRS_{k'}(\SSS',\vvv'),
    $$
    where $\SSS'= \{(a_1-\gamma)^{-1},\ldots,(a_n-\gamma)^{-1},0\}\subseteq \F_q$ and 
    $\vvv'=(v_1(a_1-\gamma)^{k},\ldots,v_{n}(a_n-\gamma)^{k},1)\in (\F_q^*)^{n+1}$. 
\end{lemma}
\begin{proof}
    With \eqref{eq.EGRS.generator matrix} and \eqref{eq.ESGRS_generator_matrix}, 
    it is easily seen that 
    \begin{align}\label{eq.nest111}
        \C_{k}(\SSS,\vvv,\infty)\subseteq \GRS_{k+1}(\SSS,\vvv,\infty).
    \end{align}
    Additionally, for any $\gamma \in \F_q\setminus \SSS$ (such $\gamma$ exists since $n\leq q-1$), 
    it follows from \cite{P1996} that  
    there exists a set $\SSS'=\{a_1',\ldots,a_{n+1}'\}\subseteq \F_q$ and 
    a vector $\vvv'=(v_1',\ldots,v_{n+1}')\in (\F_q^*)^{n+1}$ such that 
    \begin{align}\label{eq.nest222}
        \GRS_{k+1}(\SSS,\vvv,\infty)=\GRS_{k+1}(\SSS',\vvv')\subseteq \GRS_{k'}(\SSS',\vvv')
    \end{align} 
    by taking  
    $$
    a_i'=\left\{
        \begin{array}{ll}
            (a_i-\gamma)^{-1}, & {\rm if}~1\leq i\leq n, \\
            0, & {\rm if}~i=n+1,
        \end{array}
    \right.
    $$
   and 
    $$
    v_i'=\left\{
        \begin{array}{ll}
            v_i(a_i-\gamma)^{k}, & {\rm if}~1\leq i\leq n, \\
            1, & {\rm if}~i=n+1.
        \end{array}
    \right.
    $$
    Then the desired result immediately follows by combining \eqref{eq.nest111} and \eqref{eq.nest222}.
\end{proof}

\begin{remark}\label{rem.SGRS}
Note that the parameters and dual properties of SGRS codes, obtained by removing the last coordinate 
from each codeword of ESGRS codes, have been investigated in \cite{HZ2024,JMXZ2024,LELL2025}. 
Moreover, SGRS codes have been shown to be either MDS or NMDS codes \cite{HZ2024}, and in certain cases, 
they are specifically non-GRS MDS codes \cite{JMXZ2024,LELL2025}. 
From the proof of Lemma~\ref{lem.nested_relationship}, ESGRS codes may coincide with SGRS codes 
when their lengths are less than $q$ and certain additional conditions are satisfied. 
Although this possibility cannot be excluded, it is important to note that the decoding problem 
for SGRS codes also remains open. In fact, if such an equivalence holds, the decoding algorithm 
we derive in the sequel would be directly applicable to SGRS codes as well.
\end{remark}

In the following, we investigate the existence of $\ell$-ECPs of MDS ESGRS codes 
and give explicit forms of these $\ell$-ECPs if they exist.

\begin{theorem}\label{th.ECP_of_C_k}
    Suppose that $\ell=\left \lfloor \frac{d(\C_k(\SSS,{\bf v},\infty))-1}{2} \right \rfloor$. 
    If $\SSS$ is $k$-zero-sum free, then the following statements hold.  
    \begin{enumerate}
        \item [\rm 1)] If $2\nmid (n-k)$, then  $\C_k(\SSS,{\bf v},\infty)$ does not have an $\ell$-ECP.
        Moreover, $\C_k(\SSS,{\bf v},\infty)$ has an $(\ell-1)$-ECP $(\mathcal{A},\mathcal{B})$ over $\F_q$ 
            with 
            $$\mathcal{A}=\GRS_{\ell}(\SSS',{\bf 1})$$
            {\rm and}~
            $$\mathcal{B}=\GRS_{\ell-1}(\SSS',{\bf v}'^{-1}\star \uuu'),$$ 
            where    
            $\SSS'= \{(a_1-\gamma)^{-1},\ldots,(a_n-\gamma)^{-1},0\}\subseteq \F_q$,  
            $\vvv'=(v_1(a_1-\gamma)^{k},\ldots,v_{n}(a_n-\gamma)^{k},1)\in (\F_q^*)^{n+1}$ 
            and $\uuu'=(u_1',\ldots, u_{n+1}')\in (\F_q^*)^{n+1}$ with 
            \begin{align*}
                \begin{split}
                u_i' =  \left\{
                    \begin{array}{ll}
                        (a_i-\gamma)^n\prod_{1\leq j\neq i\leq n} \frac{a_j-\gamma}{a_j-a_i}, & {\rm if}~1\leq i\leq n, \\  
                        (-1)^n \prod_{1\leq j\leq n}(a_j-\gamma), & {\rm if}~i=n+1,
                    \end{array}
                \right. \\
            \end{split}
            \end{align*}
            for any $\gamma\in \F_q\setminus \SSS$.
    
        \item [\rm 2)] If $2\mid (n-k)$, then 
    $\C_k(\SSS,{\bf v},\infty)$ has an $\ell$-ECP $(\mathcal{A},\mathcal{B})$ over $\F_q$ 
        with $$\mathcal{A}=\GRS_{\ell+1}(\SSS,{\bf 1},\infty)$$ {\rm and}~ 
        $$\mathcal{B}=\GRS_{\ell}(\SSS,{\bf v}^{-1}\star \uuu, \infty)\diag(\underbrace{1,\ldots,1}_n,-1).$$
    \end{enumerate}
\end{theorem}
\begin{proof}
    1) Since $\SSS$ is $k$-zero-sum free, it follows from Lemma \ref{lem.must be MDS or NMDS}.1) that $d(\C_k(\SSS,{\bf v},\infty))=n-k+2$. 
    From $2\nmid (n-k)$ and $3\leq k\leq n-2$, we get that  
    $3\leq k\leq n-3$ and 
    $2\leq \ell=\left \lfloor \frac{d(\C_k(\SSS,{\bf v},\infty))-1}{2} \right \rfloor=\frac{n-k+1}{2}\leq \frac{n-2}{2}<\frac{n+1}{2}-1$. 
    Combining Lemmas \ref{lem.must be MDS or NMDS}.1) and \ref{lem.non-GRS}, we immediately get that 
    $\C_{k}(\SSS,{\bf v},\infty)=\C_{n-2\ell+1}(\SSS,{\bf v},\infty)$ 
    is an $[n+1,n-2\ell+1,2\ell+1]_q$ non-GRS MDS code. 
    By Lemma \ref{lem.GRS_ECP}, if $\C_{k}(\SSS,{\bf v},\infty)$ has an $\ell$-ECP, 
    then it is a GRS code, a contradiction. 
    Hence, $\C_{k}(\SSS,{\bf v},\infty)$ does not have an $\ell$-ECP.

    Moreover, when $\SSS=\F_q$, it follows from \cite[Corollary 2.8]{LW2008} 
    that $\SSS$ always contains a $k$-zero-sum subset if $3\leq k\leq n-3$. 
    Therefore, we conclude that $n\leq q-1$ and $\F_q\setminus \SSS\neq \emptyset$,  
    since $\SSS$ is assumed to be $k$-zero-sum free. 
    It then follows from Lemma \ref{lem.nested_relationship} that 
        \begin{align}\label{eq.111}
            \begin{split}
            \C_{k}(\SSS,{\bf v},\infty) & =\C_{n-2\ell+1}(\SSS,{\bf v},\infty) \\
            & \subseteq \GRS_{n-2\ell+2}(\SSS,{\bf v},\infty)\\ 
            & \subseteq \GRS_{n-2\ell+3}(\SSS',{\bf v}'),
            \end{split} 
        \end{align}
    where $\GRS_{n-2\ell+3}(\SSS',{\bf v}')$ is an $[n+1,n-2\ell+3,2\ell-1]_q$ GRS code  
    with $\SSS'= \{(a_1-\gamma)^{-1},\ldots,(a_n-\gamma)^{-1},0\}\subseteq \F_q$,    
    $\vvv'=(v_1(a_1-\gamma)^{k},\ldots,v_{n}(a_n-\gamma)^{k},1)\in (\F_q^*)^{n+1}$, 
    and any $\gamma\in \F_q\setminus \SSS$. 
    Note that $2\ell-1=2(\ell-1)+1$ and $1\leq \ell-1<\frac{n+1}{2}-1$, 
    then it follows from Lemmas \ref{lem.GRS duality}, \ref{lem.GRS_ECP}, and \ref{lem.GRS_ECP_paraA} that 
    $\GRS_{n-2\ell+3}(\SSS',{\bf v}')$ has an $(\ell-1)$-ECP $(\mathcal{A}',\mathcal{B}')$,  
    where  
    $\mathcal{A}'=\GRS_{\ell}(\SSS',{\bf r})$ is an $[n+1,\ell,n-\ell+2]_q$ MDS code 
    for some ${\bf r}\in (\F_q^*)^{n+1}$  and 
    \begin{align*}
           \mathcal{B}' & = (\GRS_{\ell}(\SSS',{\bf r})\star \GRS_{n-2\ell+3}(\SSS',{\bf v}'))^{\perp} \\
              & = \GRS_{n-\ell+2}(\SSS',{\bf r}\star {\bf v}')^{\perp} \\
              & =\GRS_{\ell-1}(\SSS',{\bf r}^{-1}\star {\bf v}^{-1}\star \uuu')
    \end{align*}
    with $\uuu'=(u_1',\ldots,u_{n+1}')$ and 
    \begin{align*}
        u_i'= & \prod_{1\leq j\neq i\leq n+1}(a_i'-a_j')^{-1} \\ 
            = & \left\{
            \begin{array}{ll}
                (a_i'-a_{n+1}')^{-1}\prod_{1\leq j\neq i\leq n}(a_i'-a_j')^{-1}, & {\rm if}~1\leq i\leq n, \\ 
                \prod_{1\leq j\leq n}(a_{n+1}'-a_j')^{-1}, & {\rm if}~i=n+1,
            \end{array}
        \right. \\
        = & \left\{
            \begin{array}{ll}
                (a_i-\gamma)^n\prod_{1\leq j\neq i\leq n}\frac{a_j-\gamma}{a_j-a_i}, & {\rm if}~1\leq i\leq n, \\ 
                (-1)^n \prod_{1\leq j\leq n}(a_j-\gamma), & {\rm if}~i=n+1.
            \end{array}
        \right. 
    \end{align*}
    Taking $\mathcal{A}={\bf r}^{-1}\star \mathcal{A}'$ and $\mathcal{B}={\bf r}\star \mathcal{B}'$,  
    it then turns out from Lemma \ref{lem.subcode_ECP} and \eqref{eq.111} that 
    \begin{align*}
    \left({\bf r}^{-1}\star \mathcal{A}', {\bf r}\star \mathcal{B}'\right) 
                              =\left(\GRS_{\ell}(\SSS',{\bf 1}), 
    \GRS_{\ell-1}(\SSS',{\bf v}'^{-1}\star \uuu') \right)
    \end{align*}
    is an $(\ell-1)$-ECP of $\C_{k}(\SSS,{\bf v},\infty)$ over $\F_q$. 
    This completes the proof of 1). 

    2) In this case, we have that 
    $1\leq \ell=\left \lfloor \frac{d(\C_k(\SSS,{\bf v},\infty))-1}{2} \right \rfloor=\frac{n-k}{2}\leq \frac{n-3}{2}<\frac{n+1}{2}-1$ 
    and $\C_{k}(\SSS,{\bf v},\infty)$ has parameters $[n+1,n-2\ell,2\ell+2]_q$. 
    Since a parity-check matrix of a linear code is a generator matrix of its dual code, 
    it then follows from \eqref{eq.ESGRS_parity_matrix} that 
    $\C_k(\SSS,{\bf v},\infty)^{\perp}$ has a generator matrix of the form 
    \begin{align}\label{eq.ESGRS_generator_matrix of dual code111}
        \begin{pmatrix}    
            \frac{u_1}{v_1}  &  \ldots & \frac{u_n}{v_n} & 0 \\ 
            \frac{u_1}{v_1}a_1    & \ldots & \frac{u_n}{v_n}a_{n} & 0 \\
    \vdots &   \ddots & \vdots & \vdots \\
    \frac{u_1}{v_1}a_1^{2\ell-2}    & \ldots & \frac{u_n}{v_n}a_{n}^{2\ell-2} & 0 \\
    \frac{u_1}{v_1}a_1^{2\ell-1}    & \ldots & \frac{u_n}{v_n}a_{n}^{2\ell-1} & -1 \\
    \frac{u_1}{v_1}a_1^{2\ell}    &  \ldots &  \frac{u_n}{v_n}a_{n}^{2\ell} & -\sum_{i=1}^na_i
        \end{pmatrix}.    
    \end{align}
    Take $\mathcal{A}=\GRS_{\ell+1}(\SSS,{\bf 1},\infty)$ and 
    $\mathcal{B}=\GRS_{n+1-\ell}(\SSS,{\bf v},\infty)^{\perp}$. 
    It follows from Lemma \ref{lem.GRS duality}.2) and \eqref{eq.EGRS.generator matrix} that 
    $\mathcal{B}=\GRS_{\ell}(\SSS,{\bf v}^{-1}\star \uuu,\infty)\diag(\underbrace{1,\ldots,1}_n,-1)$, and 
    $\mathcal{A}$ and $\mathcal{B}$ have generator matrices of the forms 
    \begin{align*}
        G_{\mathcal{A}}=\begin{pmatrix}    
            1 &   \ldots & 1 & 0 \\ 
            a_1  &  \ldots & a_{n} & 0 \\
            \vdots  &  \ddots & \vdots & \vdots \\
            a^{\ell-1}_1  &  \ldots & a^{\ell-1}_{n} & 0 \\
            a^{\ell}_1 &   \ldots & a^{\ell}_{n} & 1
        \end{pmatrix}
    \end{align*} 
    and 
    \begin{align*}
        G_{\mathcal{B}}=\begin{pmatrix}    
            \frac{u_1}{v_1}  &  \ldots & \frac{u_n}{v_n} & 0 \\ 
            \frac{u_1}{v_1}a_1    & \ldots & \frac{u_n}{v_n}a_{n} & 0 \\
    \vdots &   \ddots & \vdots & \vdots \\
    \frac{u_1}{v_1}a_1^{\ell-2}    & \ldots & \frac{u_n}{v_n}a_{n}^{\ell-2} & 0 \\
    \frac{u_1}{v_1}a_1^{\ell-1}    & \ldots & \frac{u_n}{v_n}a_{n}^{\ell-1} & -1 \\
        \end{pmatrix},
    \end{align*}
    respectively.

    It turns out that $\mathcal{A}\star \mathcal{B}$ is an $[n+1,2\ell]_q$ linear code 
    with generator matrix of the form 
    \begin{align*}
        G_{\mathcal{A}\star \mathcal{B}}=\left(
            \begin{array}{ccccc}
                \frac{u_1}{v_1}  &  \ldots & \frac{u_n}{v_n} & 0 \\ 
                \frac{u_1}{v_1}a_1    & \ldots & \frac{u_n}{v_n}a_{n} & 0 \\
        \vdots &   \ddots & \vdots & \vdots \\
        \frac{u_1}{v_1}a_1^{2\ell-2}    & \ldots & \frac{u_n}{v_n}a_{n}^{2\ell-2} & 0 \\
        \frac{u_1}{v_1}a_1^{2\ell-1}    & \ldots & \frac{u_n}{v_n}a_{n}^{2\ell-1} & -1 \\              
            \end{array}
        \right),    
\end{align*}
which, by combining with \eqref{eq.ESGRS_generator_matrix of dual code111}, 
implies that $\mathcal{A}\star \mathcal{B}\subseteq \C_k(\SSS,{\bf v},\infty)^{\perp}$, $i.e.,$ the condition \rmnum{1}) 
in Definition \ref{def.ECP} holds. 
In addition, it can be checked that 
$$d(\mathcal{B}^{\perp})=d(\GRS_{n+1-\ell}(\SSS,{\bf v},\infty))=\ell+1>\ell,$$  
$$\dim(\mathcal{A})=\ell+1>\ell,$$
{\rm and} 
\begin{align*}
d(\mathcal{A})+d(\C_{k}(\SSS,{\bf v},\infty)) & = n+1-(\ell+1)+1+2\ell+2 \\ 
                                              & = n+\ell+3 \\
                                              & > n+1, 
\end{align*}
which implies that the conditions \rmnum{2}), \rmnum{3}) and \rmnum{4}) in Definition \ref{def.ECP} also hold. 
With Definition \ref{def.ECP}, 
we immediately conclude that 
$$
\left(
    \GRS_{\ell+1}(\SSS,{\bf 1},\infty),
\GRS_{\ell}(\SSS,{\bf v}^{-1}\star \uuu,\infty)\diag(\underbrace{1,\ldots,1}_n,-1)
\right)$$ 
is an $\ell$-ECP of $\C_k(\SSS,{\bf v},\infty)$ over $\F_q$. 
This completes the proof of 2). 
\end{proof}

\subsection{The case of NMDS ESGRS codes}

This subsection focuses on NMDS ESGRS codes. 
In this case, the existence of $\ell$-ECPs is always guaranteed. 
The following theorem presents the explicit forms of their $\ell$-ECPs, 
which are constructed based on suitable variants of GRS codes.

\begin{theorem}\label{th.ECP_of_C_k222}
    Suppose that $\ell=\left \lfloor \frac{d(\C_k(\SSS,{\bf v},\infty))-1}{2} \right \rfloor$. 
    If $\SSS$ contains a $k$-zero-sum subset, then the following statements hold. 
    \begin{enumerate}
        \item [\rm 1)] If $2\nmid (n-k)$, then $\C_k(\SSS,{\bf v},\infty)$ has an $\ell$-ECP $(\mathcal{A},\mathcal{B})$ 
        over $\F_q$ with 
        $$\mathcal{A}=\{(f(a_1),\ldots,f(a_n),0):~f(x)\in \F_{q}[x]_{\ell+1}\}$$
        {\rm and}~ 
        $$\mathcal{B}=\GRS_{\ell}(\SSS,{\bf v}^{-1}\star \uuu, \infty).$$

        \item [\rm 2)] If $2\mid (n-k)$, then $\C_k(\SSS,{\bf v},\infty)$ has an $\ell$-ECP $(\mathcal{A},\mathcal{B})$ 
        over $\F_q$ with 
        $$\mathcal{A}=\GRS_{\ell+1}(\SSS,{\bf 1},\infty)$$ 
        {\rm and}~
        $$\mathcal{B}=\GRS_{\ell}(\SSS,{\bf v}^{-1}\star \uuu, \infty)\diag(\underbrace{1,\ldots,1}_n,-1).$$
    \end{enumerate}
\end{theorem}
\begin{proof} 
    1) Since $\SSS$ contains a $k$-zero-sum subset,  
    it follows from Lemma \ref{lem.must be MDS or NMDS}.2) that $d(\C_k(\SSS,{\bf v},\infty))=n-k+1$.  
    From $2\nmid (n-k)$ and $3\leq k\leq n-2$, we get that  
    $3\leq k\leq n-3$ and $1\leq \ell=\left \lfloor \frac{d(\C_k(\SSS,{\bf v},\infty))-1}{2} \right \rfloor=\frac{n-k-1}{2}\leq \frac{n-4}{2}<\frac{n+1}{2}-1$ 
    and $\C_{k}(\SSS,{\bf v},\infty)$ has parameters $[n+1,n-2\ell-1,2\ell+2]_q$. 
    In this case, we immediately derive from \eqref{eq.ESGRS_parity_matrix} that 
    $\C_k(\SSS,{\bf v},\infty)^{\perp}$ has a generator matrix of the form 
    \begin{align}\label{eq.ESGRS_generator_matrix of dual code222}
        \begin{pmatrix}    
            \frac{u_1}{v_1}  &  \ldots & \frac{u_n}{v_n} & 0 \\ 
            \frac{u_1}{v_1}a_1    & \ldots & \frac{u_n}{v_n}a_{n} & 0 \\
    \vdots &   \ddots & \vdots & \vdots \\
    \frac{u_1}{v_1}a_1^{2\ell-1}    & \ldots & \frac{u_n}{v_n}a_{n}^{2\ell-1} & 0 \\
    \frac{u_1}{v_1}a_1^{2\ell}    & \ldots & \frac{u_n}{v_n}a_{n}^{2\ell} & -1 \\
    \frac{u_1}{v_1}a_1^{2\ell+1}    &  \ldots &  \frac{u_n}{v_n}a_{n}^{2\ell+1} & -\sum_{i=1}^na_i
        \end{pmatrix}.    
    \end{align}
    Take 
    \begin{align*}
        \mathcal{A} & = \{(c_1,\ldots,c_n,0):~(c_1,\ldots,c_n)\in \GRS_{\ell+1}(\SSS,{\bf 1})\} \\
                    & = \{(f(a_1),\ldots,f(a_n),0):~f(x)\in \F_{q}[x]_{\ell+1}\}
    \end{align*}
    and 
    $\mathcal{B}=\GRS_{\ell}(\SSS,{\bf v}^{-1}\star \uuu,\infty)$. 
    According to \eqref{eq.GRS.generator matrix}, we deduce that  
    $\mathcal{A}$ and $\mathcal{B}$ have generator matrices of the forms 
    \begin{align*}
        G_{\mathcal{A}}=\begin{pmatrix}    
            1 &   \ldots & 1 & 0 \\ 
            a_1  &  \ldots & a_{n} & 0 \\
            \vdots  &  \ddots & \vdots & \vdots \\
            a^{\ell}_1  &  \ldots & a^{\ell}_{n} & 0
        \end{pmatrix}
    \end{align*} 
    and 
    \begin{align*}    
        G_{\mathcal{B}}=\begin{pmatrix}    
            \frac{u_1}{v_1}  &  \ldots & \frac{u_n}{v_n} & 0 \\ 
            \frac{u_1}{v_1}a_1    & \ldots & \frac{u_n}{v_n}a_{n} & 0 \\
    \vdots &   \ddots & \vdots & \vdots \\
    \frac{u_1}{v_1}a_1^{\ell-2}    & \ldots & \frac{u_n}{v_n}a_{n}^{\ell-2} & 0 \\
    \frac{u_1}{v_1}a_1^{\ell-1}    & \ldots & \frac{u_n}{v_n}a_{n}^{\ell-1} & 1 \\
        \end{pmatrix},
    \end{align*}
    respectively.

    It is not difficult to show that $\mathcal{A}\star \mathcal{B}$ is an $[n+1,2\ell]_q$ linear code 
    with generator matrix of the form 
    \begin{align*}
        G_{\mathcal{A}\star \mathcal{B}}=\left(
            \begin{array}{cccc}
                \frac{u_1}{v_1}  &  \ldots & \frac{u_n}{v_n} & 0 \\ 
                \frac{u_1}{v_1}a_1    & \ldots & \frac{u_n}{v_n}a_{n} & 0 \\
        \vdots &  \ddots & \vdots & \vdots \\
        \frac{u_1}{v_1}a_1^{2\ell-1}   & \ldots & \frac{u_n}{v_n}a_{n}^{2\ell-1} & 0 \\              
            \end{array}
        \right).     
\end{align*}
Combining with \eqref{eq.ESGRS_generator_matrix of dual code222}, 
we get that $\mathcal{A}\star \mathcal{B}\subseteq \C_k(\SSS,{\bf v},\infty)^{\perp}$ 
and hence, the condition \rmnum{1}) in Definition \ref{def.ECP} holds. 
Moreover, since the dual code of an MDS code is still MDS, we can also verify that  
$$d(\mathcal{B}^{\perp})=d(\GRS_{\ell}(\SSS,{\bf v}^{-1}\star \uuu,\infty)^{\perp})=\ell+1>\ell$$
$$\dim(\mathcal{A})=\ell+1>\ell,$$ 
{\rm and} 
\begin{align*}
d(\mathcal{A})+d(\C_{k}(\SSS,{\bf v},\infty)) & = n+1-(\ell+1)+2\ell+2 \\ 
                                              & = n+\ell+2 \\
                                              & > n+1,
\end{align*}
Therefore, the conditions \rmnum{2}), \rmnum{3}) and \rmnum{4}) in Definition \ref{def.ECP} also hold. 
With Definition \ref{def.ECP}, 
we immediately conclude that 
$$\left( \{(f(a_1),\ldots,f(a_n),0):~f(x)\in \F_{q}[x]_{\ell+1}\},
\GRS_{\ell}(\SSS,{\bf v}^{-1}\star \uuu,\infty) \right)$$ 
is an $\ell$-ECP of $\C_k(\SSS,{\bf v},\infty)$ over $\F_q$. 
This completes the proof of 1). 

2) In this case, we have that 
$1\leq \ell=\left \lfloor \frac{d(\C_k(\SSS,{\bf v},\infty))-1}{2} \right \rfloor=\frac{n-k}{2}\leq \frac{n-3}{2}<\frac{n+1}{2}-1$ 
and $\C_{k}(\SSS,{\bf v},\infty)$ has parameters $[n+1,n-2\ell,2\ell+1]_q$. 
Take $\mathcal{A}=\GRS_{\ell+1}(\SSS,{\bf 1},\infty)$ and 
$\mathcal{B}=\GRS_{\ell}(\SSS,{\bf v}^{-1}\star \uuu, \infty)\diag(\underbrace{1,\ldots,1}_n,-1)$ 
as Theorem \ref{th.ECP_of_C_k}.2). 
By arguments similar to those in the proofs of Theorem \ref{th.ECP_of_C_k}.2) and 1) above, 
it suffices to show that 
\begin{align*}
d(\mathcal{A})+d(\C_{k}(\SSS,{\bf v},\infty)) & = n+1-(\ell+1)+1+2\ell+1\\ & = n+\ell+2 \\ &> n+1.
\end{align*}
As a result, we conclude that 
$$
\left(
    \GRS_{\ell+1}(\SSS,{\bf 1},\infty), \GRS_{\ell}(\SSS,{\bf v}^{-1}\star \uuu, \infty)\diag(\underbrace{1,\ldots,1}_n,-1)
\right)$$
is an $\ell$-ECP of $\C_k(\SSS,{\bf v},\infty)$ over $\F_q$  
by Definition \ref{def.ECP}.
This completes the proof of 2). 
\end{proof}

\begin{remark}
    Note that we suppose $\ell=\left \lfloor \frac{d(\C_k(\SSS,{\bf v},\infty))-1}{2} \right \rfloor$ 
    for a more convenient representation in Theorems \ref{th.ECP_of_C_k} and \ref{th.ECP_of_C_k222}. 
    Combining with Lemma \ref{lem.must be MDS or NMDS}, we can further deduce the explicit values of $\ell$ 
    as listed in Table \ref{tab.ECP}, where the notation is the same with that used in 
    Theorems \ref{th.ECP_of_C_k} and \ref{th.ECP_of_C_k222}. 
    
    \begin{table*}[h!]
        \centering
        \caption{{$\ell$-ECPs $(\mathcal{A},\mathcal{B})$ of ESGRS codes $\C_{k}(\SSS,{\bf v},\infty)$}}\label{tab.ECP}       
        \scalebox{0.85}{
            \begin{tabular}{c|c|c|c|c}
         \hline 
        $\C_k(\SSS,{\bf v},\infty)$ & $2\mid (n-k)?$  & $\ell$ & ECP $(\mathcal{A},\mathcal{B})$ & Result \\ \hline  \hline
        $[n+1,k,n-k+2]_q$ & No  &  $\frac{n-k-1}{2}$ 
        & $\left(\GRS_{\frac{n-k+1}{2}}(\SSS',{\bf 1}),\GRS_{\frac{n-k-1}{2}}(\SSS',{\bf v}'^{-1}\star {\bf u}')\right)$ & Theorem \ref{th.ECP_of_C_k}.1)\\
        \hline

        $[n+1,k,n-k+2]_q$ & Yes  & $\frac{n-k}{2}$ & 
        $\left(\GRS_{\frac{n-k+2}{2}}(\SSS,{\bf 1},\infty),\GRS_{\frac{n-k}{2}}(\SSS,{\bf v}^{-1}\star {\bf u},\infty)\diag(\underbrace{1,\ldots,1}_n,-1)\right)$ & Theorem \ref{th.ECP_of_C_k}.2)\\ 
        \hline

        $[n+1,k,n-k+1]_q$ & No  &  $\frac{n-k-1}{2}$ & 
        $\left(\{(f(a_1),f(a_2),\ldots,f(a_n),0):~f(x)\in \F_q[x]_{\frac{n-k+1}{2}}\},\GRS_{\frac{n-k-1}{2}}(\SSS,{\bf v}^{-1}\star {\bf u},\infty)\right)$ & Theorem \ref{th.ECP_of_C_k222}.1)\\
        \hline

        $[n+1,k,n-k+1]_q$ & Yes  & $\frac{n-k}{2}$ & 
        $\left(\GRS_{\frac{n-k+2}{2}}(\SSS,{\bf 1},\infty),\GRS_{\frac{n-k}{2}}(\SSS,{\bf v}^{-1}\star {\bf u},\infty)\diag(\underbrace{1,\ldots,1}_n,-1)\right)$ & Theorem \ref{th.ECP_of_C_k222}.2) \\
                \hline
        \end{tabular}}
    \end{table*} 
\end{remark}

\section{Unique decoding of ESGRS codes based on their ECPs}\label{sec.4}

In this section, we present an explicit unique decoding algorithm for ESGRS codes based on their $\ell$-ECPs 
established in Section \ref{sec.3} and discuss its decoding complexities. 
We also present some concrete examples to illustrate our decoding algorithm. 

\subsection{An explicit unique decoding algorithm of ESGRS codes}

In this subsection, we provide an explicit unique decoding algorithm for ESGRS codes 
$\C_k(\SSS,{\bf v},\infty)$ based on their $\ell$-ECPs. 
First of all, we recall the error-correcting mechanism of $\ell$-ECPs for a given $[n,k]_q$ linear code $\C$ 
(see also \cite{HL2023DM,CP2020}).  
Let $(\mathcal{A},\mathcal{B})$ be an $\ell$-ECP of $\C$ and let  
${\bf y}={\bf c}+{\bf e}\in \F_q^n$ be a received vector with ${\bf c}\in \C$ and $\wt({\bf e})\leq \ell$. 
Note that based on conditions ${\rm \rmnum{1})}$-${\rm \rmnum{3})}$ in Definition \ref{def.ECP}, 
Duursma and K\"otter showed in \cite[Theorem 1]{DK1994} that 
there exists a nonzero vector ${\bf a}\in \mathcal{A}$ such that 
    \begin{align*}
        \langle ({\bf a}\star {\bf b}), {\bf y}\rangle=0,~\forall~ {\bf b}\in \mathcal{B}
    \end{align*}
    or equivalently, 
    \begin{align}\label{eq.ECPDM222}
        G_{\mathcal{B}}\cdot \diag({\bf y}) \cdot {\bf a}^T={\bf 0}, 
    \end{align}
where $G_{\mathcal{B}}$ denotes a generator matrix of $\mathcal{B}$ and 
$\diag({\bf y})$ represents the $n\times n$ diagonal matrix that has the elements of ${\bf y}$ on its main diagonal, 
and any solution ${\bf a}\in \mathcal{A}$ satisfies ${\bf e}\star {\bf a}={\bf 0}$.  
Then it is easily seen that $\supp({\bf e})\subseteq z({\bf a})$, and hence, 
any nonzero solution ${\bf a}$ can efficiently locate the possible error positions of ${\bf y}$. 

In addition, Pellikaan pointed out in \cite{P1992} that the condition 
``${\rm \rmnum{4})}$ $d(\mathcal{A})+d(\C)> n$'' in Definition \ref{def.ECP} 
can guarantee  that 
\begin{align}\label{eq.ECPDM333}
    H_\C{\bf x}^T = H_\C{\bf y}^T=H_\C{\bf e}^T
\end{align} 
has a unique solution, 
where $H_{\C}$ is a parity-check matrix of $\C$. 
It then follows that ${\bf x}={\bf e}$ and $x_i = 0$ for all $i \notin z({\bf a})$ 
if we write ${\bf x}=(x_1,\ldots,x_n)$.
As a result, the $\ell$-ECP $(\mathcal{A},\mathcal{B})$ has uniquely decoded the received vector ${\bf y}$ 
to the codeword ${\bf c}={\bf y}-{\bf x}$.

Based on the above discussion, we conclude that to decode a linear code $\C$ using its $\ell$-ECP $(\mathcal{A}, \mathcal{B})$, 
one should first determine explicit representations of the code pair $\mathcal{A}$ and $\mathcal{B}$, 
and then solve the system of equations \eqref{eq.ECPDM222} and \eqref{eq.ECPDM333}. 
By combining Theorems \ref{th.ECP_of_C_k} and \ref{th.ECP_of_C_k222}, 
we can now present a decoding algorithm for $\C_k(\SSS,{\bf v},\infty)$ based on its $\ell$-ECP as follows.

\begin{theorem}\label{th.decoding ECP}
    If $\SSS$ is $k$-zero-sum free and $2\nmid (n-k)$, 
    then an ESGRS code $\C_k(\SSS,{\bf v},\infty)$ can be uniquely decoded by Algorithm \ref{alg.1} in $O(n^3)$ time. 
    Moreover, the decoding algorithms with the same time complexity for $\C_k(\SSS,{\bf v},\infty)$ with respect to 
    the other three cases discussed in Theorems \ref{th.ECP_of_C_k} and \ref{th.ECP_of_C_k222} 
    can be similarly derived by modifying Lines 8-14 of Algorithm \ref{alg.1} based on Table \ref{tab.ECP} accordingly. 
\end{theorem}
\begin{proof}
    The unique decoding correctness is straightforward by combining Theorems \ref{th.ECP_of_C_k} and \ref{th.ECP_of_C_k222} as well as the discussions above. 
    By \cite[Remark 2.2(1)]{P1996}, the algorithm takes $O((n t)^3)$ time, with $t$ denoting the degree of the field extension 
of the finite field that contains the $\ell$-ECP $(\mathcal{A},\mathcal{B})$ over the base field of the ESGRS code $\C_k(\SSS,{\bf v},\infty)$. 
By Theorems \ref{th.ECP_of_C_k} and \ref{th.ECP_of_C_k222}, we have $t=1$ and hence, 
the time complexity reduces to $O(n^3)$. 
    More details can be found in Algorithm \ref{alg.1}.
\end{proof}

\subsection{More discussions on the decoding complexity}\label{subsec.complexity}
 
In Theorem \ref{th.decoding ECP}, we showed that the time complexity of Algorithm \ref{alg.1} is $O(n^3)$ in general.
Similar to \cite{JYS2025}, \cite{SY2023}, \cite{SYJL2024}, and \cite{WLL2025}, the extended Euclidean algorithm could also be 
a potential method to further deduce a unique decoding algorithm for $\C_k(\SSS,{\bf v},\infty)$ that operates in $O(qn)$ time. 
Since our algorithm complexity only depends on the code length $n+1$ and $n\leq q$, 
it is more efficient than those based on the extended Euclidean algorithm when $q>n^{2+\epsilon}$ for any $\epsilon>0$. 
On the other hand, the extended Euclidean algorithm starts from the parity-check matrix, but according to Lemma \ref{lem.parity check matrix}, 
the dual code of an ESGRS code may not be an ESGRS code of the same form. 
Therefore, even when the code length is linearly related to $q$, additional work must be done to precisely 
implement the extended Euclidean algorithm. 
In this sense, our decoding algorithm based on $\ell$-ECPs can achieve better performance than the extended Euclidean algorithm.

Moreover, focusing on the system of linear equations
$$G_{\mathcal{B}}\cdot \diag({\bf y}) \cdot G_{\mathcal{A}}^T\cdot {\bf s}^T={\bf 0},$$
which appears in Line 15 of Algorithm \ref{alg.1}, we know from \cite[Section \Rmnum{6}.B]{DK1994} that 
there are two potential approaches to reduce the decoding complexity of Algorithm \ref{alg.1} to $O(n^2)$ or even lower: 
\begin{itemize}
    \item exploiting the possible regular structures in the matrix $G_{\mathcal{B}}\cdot \diag({\bf y}) \cdot G_{\mathcal{A}}^T$;
    \item reducing the size of the field that contains the entries of the matrix $G_{\mathcal{B}}\cdot \diag({\bf y}) \cdot G_{\mathcal{A}}^T$.
\end{itemize}
According to Theorems \ref{th.ECP_of_C_k} and \ref{th.ECP_of_C_k222}, the matrices $G_{\mathcal{A}}$ and $G_{\mathcal{B}}$ 
are closely related to the choices of $\SSS$ and ${\bf v}$, rather than anything else.  
Note also that $\diag({\bf y})$ is a diagonal matrix only depending on the error vector 
${\bf e}$ with ${\bf y}={\bf c}+{\bf e}$.  
Therefore, it is possible to design $\SSS$ and ${\bf v}$ or consider some special error patterns 
such that the above two approaches can be reached, leading to more efficient computations of solving the system of linear equations 
$G_{\mathcal{B}}\cdot \diag({\bf y}) \cdot G_{\mathcal{A}}^T\cdot {\bf s}^T={\bf 0}.$

\begin{algorithm}\label{alg.1}
    \caption{An algorithm for uniquely decoding $\C_k(\SSS,\vvv,\infty)$ 
    when $\SSS$ is $k$-zero-sum free and $2\nmid (n-k)$}
    \KwIn{A received vector ${\bf y}=(y_1,\ldots,y_{n+1})\in \F_q^{n+1}$}
    \KwOut{A decoded codeword ${\bf c}\in \C_k(\SSS,\vvv,\infty)$ or a failure message}
    \Begin{
        $H_{\C_k(\SSS,\vvv,\infty)}\leftarrow$ the parity-check matrix of $\C_k(\SSS,\vvv,\infty)$ given by \eqref{eq.ESGRS_generator_matrix}\\
        \uIf {$H_{\C_k(\SSS,\vvv,\infty)}\cdot {\bf y}^T={\bf 0}$} 
                {${\bf c} \leftarrow {\bf y}$\\
                \Return ${\bf c}$ is the decoded codeword} 
        \Else {
            $\diag({\bf y})\leftarrow \diag(y_1,\ldots,y_{n+1})$\\
            $\ell \leftarrow \frac{n-k-1}{2}$\\
            $\gamma\leftarrow$ arbitrary nonzero element of $\F_q\setminus \SSS$\\
            $\SSS'\leftarrow \{(a_1-\gamma)^{-1},\ldots,(a_n-\gamma)^{-1},0\}$\\
            $\vvv'\leftarrow (v_1(a_1-\gamma)^{k-1},\ldots,v_{n}(a_n-\gamma)^{k-1},1)$\\
            $\uuu'\leftarrow (u_1',\ldots,u_{n+1}')$ with $u_i'$ given in Theorem \ref{th.ECP_of_C_k}.1) \\
            $G_{\mathcal{A}}\leftarrow $ the generator matrix of $\GRS_{\ell+1}(\SSS',{\bf 1})$ given by \eqref{eq.GRS.generator matrix}\\
            $G_{\mathcal{B}}\leftarrow $ the generator matrix of $\GRS_{\ell}(\SSS',\vvv'^{-1}\star \uuu')$ given by \eqref{eq.GRS.generator matrix}\\
            ${\bf s}({\bf y}) \leftarrow \{{\bf s}=(s_1,\ldots,s_{\dim(\mathcal{A})}):~ G_{\mathcal{B}}\cdot \diag({\bf y}) \cdot G_{\mathcal{A}}^T\cdot {\bf s}^T={\bf 0}\}$\\

            \uIf {${\bf s}({\bf y})=\{{\bf 0}\}$} 
                {\Return ${\bf y}$ has more than $\ell$ errors} 
        \Else {
            ${\bf s}_0\leftarrow $ any nonzero vector in ${\bf s}({\bf y})$\\
            ${\bf a}\leftarrow {\bf s}_0 G_\mathcal{A}$\\
            $Z\leftarrow z({\bf a})$\\
            ${\bf x}({\bf y})\leftarrow \{{\bf x}:~ 
            H_{\C_k(\SSS,\vvv,\infty)}\cdot {\bf x}^T = H_{\C_k(\SSS,\vvv,\infty)}\cdot {\bf y}^T~{\rm with}~
            x_i=0,~\forall~i\notin Z\}$\\

            \uIf {${\bf x}({\bf y})=\{{\bf x}_0\}$ and $\wt({\bf x}_0)\leq \ell$} 
                {${\bf c}\leftarrow {\bf y}-{\bf x}_0$\\
                \Return ${\bf c}$ is the decoded codeword}
        \Else {
            \Return ${\bf y}$ has more than $\ell$ errors

            } 
            }

        }
           
}
\end{algorithm}

\subsection{Decoding examples of success and failure} 

Finally, we give some specific examples to illustrate the decoding procedure given in Algorithm \ref{alg.1}.
\begin{example}\label{exam.1}
    Let $q=17$, $\SSS=\{1,3,5,7,10,12,14,16\}\subseteq \F_{17}$, and $\vvv=(1,1,1,1,1,1,1,1)\in (\F_{17}^*)^8$. 
    Consider the ESGRS code $\C_3(\SSS,\vvv,\infty)$ 
    with generator matrix 
    \begin{equation*}
        G_3=\begin{pmatrix}    
    1 &  1 & 1 & 1 & 1 & 1 & 1 & 1 & 0 \\
    1 &  3 & 5 & 7 & 10 & 12 & 14 & 16 & 0 \\ 
    1 &  10 & 6 & 3 & 14 & 11 & 7 & 16 & 1
        \end{pmatrix}. 
    \end{equation*}
    Suppose that the received vector is ${\bf y}=(4, 6, 1, 14, 5, 7,  12, 15, \\ 2)\in \F_{17}^9$ and 
    we have $$\diag({\bf y})=\diag(4, 6, 1, 14, 5, 7, 12, 15, 2).$$
    We have the following results. 
    \begin{itemize}
        \item Note that $\SSS$ is $3$-zero-sum free.
        It then follows from Lemmas \ref{lem.parity check matrix}, \ref{lem.must be MDS or NMDS}.1), 
        and \ref{lem.non-GRS} that the ESGRS code $\C_3(\SSS,\vvv,\infty)$ is a 
        $[9,3,7]_{17}$ non-GRS MDS code with parity-check matrix 
        \begin{equation*}
            H_{\C_3(\SSS,\vvv,\infty)}=\begin{pmatrix}    
        4 &  1 & 11 & 13 & 4 & 6 & 16 & 13 & 0 \\
        4 &  3 & 4 & 6 & 6 & 4 & 3 & 4 & 0 \\ 
        4 &  9 & 3 & 8 & 9 & 14 & 8 & 13 & 0 \\
        4 &  10 & 15 & 5 & 5 & 15 & 10 & 4 & 0 \\
        4 &  13 & 7 & 1 & 16 & 10 & 4 & 13 & 16 \\
        4 &  5 & 1 & 7 & 7 & 1 & 5 & 4 & 0 
            \end{pmatrix}. 
        \end{equation*} 

        Note also that $$H_{\C_3(\SSS,\vvv,\infty)}\cdot {\bf y}^T=(1,4,7,12,13,1)^T\neq {\bf 0},$$ 
        so ${\bf y}$ is not a codeword of $\C_3(\SSS,\vvv,\infty)$.

        \item From Lines 8 to 14 of Algorithm \ref{alg.1}, we have $\ell=2$ and 
        we can take 
        $\gamma=2\in \F_{17}\setminus \SSS$,  
        $\SSS'=\{16,1,6,7,15,12,10,11,0\}$, 
        $\vvv'=( 1, 1, 9, 8, 13, 15, 8, 9,1)$, and 
        $\uuu'=(12, 14, 11, 1, 7, 12, 4, 4, 3 )$. 
        Then $(\mathcal{A},\mathcal{B})$ is a $2$-ECP of $\C_3(\SSS,\vvv,\infty)$, 
        where  
        $$
        \mathcal{A}=\GRS_{3}(\SSS',{\bf 1})~\mbox{and}~\mathcal{B}=\GRS_{2}(\SSS',{\bf v}'^{-1}\star {\bf u}').
        $$
        Moreover, we deduce that 
        $$
        G_{\mathcal{A}}=\begin{pmatrix} 
            1 &  1 & 1 & 1 & 1 & 1 & 1 & 1 & 1 \\   
            16 & 1 & 6 & 7 & 15 & 12 & 10 & 11 & 0 \\
            1 &  1 & 2 & 15 & 4 & 8 & 15 & 2 & 0 \\   
                \end{pmatrix}
        $$
        \mbox{and}~
        $$
        G_{\mathcal{B}}=\begin{pmatrix}
            5 & 14 & 13 & 3 & 12 & 13 & 5 & 3 & 3 \\
            12 &  14 & 10 & 4 & 10 & 3 & 16 & 16 & 0      
        \end{pmatrix}.
        $$
        
       \item With Line 15 of Algorithm \ref{alg.1}, we solve the system of linear equations 
       $$G_{\mathcal{B}}\cdot \diag({\bf y}) \cdot G_{\mathcal{A}}^T\cdot {\bf s}^T={\bf 0}$$
       and get ${\bf s}({\bf y})=\{\mu(1,6,10):~ \mu \in \F_{17}\}$.  
       Take ${\bf s}_0=(1,6,10)$, then we have  
       $${\bf a}={\bf s}_0 G_{\mathcal{A}}=(5,0,6,6,12,0,7,2,1)~{\rm and}~Z=z({\bf a})=\{2,6\}.$$

       \item Based on Line 22 of Algorithm \ref{alg.1}, we further set ${\bf x}=(0,x_2,0,0,0,x_6,0,0,0)$ and 
       solve the system of linear equations  
       $$H_{\C_3(\SSS,\vvv,\infty)}{\bf x}^T=H_{\C_3(\SSS,\vvv,\infty)}{\bf y}^T.$$ 
       As a result, we get the solution ${\bf x}_0=\{0,16,0,0,0,6,0,0,\\0\}$ and $\wt({\bf x}_0)=2$. 

       \item Finally, according to Lines 24 and 25 of Algorithm \ref{alg.1}, 
       the received vector ${\bf y}$ is then uniquely decoded to 
       the codeword 
       $${\bf c}={\bf y}-{\bf x}_0=(4, 7, 1, 14, 5, 1, 12, 15, 2).$$ 
       Note also that $H_{\C_3(\SSS,\vvv,\infty)}\cdot {\bf c}^T={\bf 0}$. 
       Hence, the decoding process is successful. 
    \end{itemize}

\end{example}

\begin{example}\label{exam.2}
    Let $q=16$ and let $\w$ be a primitive element of $\F_{16}$ such that $\w^4+\w+1=0$, 
    $\SSS=\{1,\w,\w^2,\w^3,\w^4,\w^6, \w^7,\w^8,\w^9,\w^{11},\w^{12},\w^{13},\w^{14}\}\subseteq \F_{16}$, 
    and $\vvv=(1,1,1,1,1,1,1,1,1,1,1,1,1)\in (\F_{16}^*)^{13}$. 
    Consider the ESGRS code $\C_7(\SSS,\vvv,\infty)$ 
    with generator matrix of the form \eqref{eq.G7}.  
\begin{figure*}[!t]
\begin{equation}\label{eq.G7}
        G_7=\left(
            \begin{array}{cccccccccccccc}
                1 &  1 & 1 & 1 & 1 & 1 & 1 & 1 & 1 &  1 & 1 & 1 & 1 & 0 \\
                1 & \w & \w^2 & \w^3 & \w^4 & \w^6 & \w^7 & \w^8 & \w^9 & \w^{11} & \w^{12} & \w^{13} & \w^{14} & 0 \\
                1 & \w^2 & \w^4 & \w^6 & \w^8 & \w^{12} & \w^{14} & \w & \w^3 & \w^{7} & \w^{9} & \w^{11} & \w^{13} & 0 \\
                1 & \w^3 & \w^6 & \w^9 & \w^{12} & \w^{3} & \w^{6} & \w^{9} & \w^{12} & \w^{3} & \w^{6} & \w^{9} & \w^{12} & 0 \\
                1 & \w^4 & \w^8 & \w^{12} & \w & \w^9 & \w^{13} & \w^2 & \w^6 & \w^{14} & \w^{3} & \w^7 & \w^{11} & 0 \\
                1 & \w^5 & \w^{10} & 1 & \w^5 & 1 & \w^{5} & \w^{10} & 1 & \w^{10} & 1 & \w^5 & \w^{10} & 0 \\
                1 & \w^7 & \w^{14} & \w^6 & \w^{13} & \w^{12} & \w^4 & \w^{11} & \w^3 & \w^2 & \w^{9} & \w & \w^{8} & 1 \\              
            \end{array}
        \right)
\end{equation}   
    \hrulefill
\end{figure*}
    Suppose that the received vector is 
    ${\bf y}=(\w^{12},\w, \w^{13},\w^7,1,\w^3,\w^9,\w,\w^2,\w,\w^7,\w^6,\w^{12},\w^{11})\in \F_{16}^{14}$ and 
    we have 
    \begin{align*}
        \diag({\bf y})=\diag(\w^{12},\w, & \w^{13},\w^7,1, \w^3, \\ & \w^9,\w,\w^2,\w,\w^7,\w^6,\w^{12},\w^{11}).
    \end{align*}
    We have the following results.  
    \begin{itemize}
        \item Note that $\SSS$ contains a $7$-zero-sum subset, 
        such as the $7$-subset $\{1, \w^{12}, \w^2, \w^{13}, \w^3, \w^8, \w^9\}$. 
        It then follows from Lemmas \ref{lem.parity check matrix} and \ref{lem.must be MDS or NMDS}.2) 
        that the ESGRS code $\C_7(\SSS,\vvv,\infty)$ is a 
        $[14,7,7]_{16}$ NMDS code with parity-check matrix of the form \eqref{eq.H7}. 
\begin{figure*}[!t]
\begin{equation}\label{eq.H7}
            H_{\C_7(\SSS,\vvv,\infty)}=
            \left(
                \begin{array}{cccccccccccccc}
                    1 &  \w^{11} & \w^7 & \w^{11} & \w^{14} & \w^7 & \w^{11} & \w^{13} & \w^{13} & \w^{13} & \w^{14} & \w^{14} & \w^7 & 0 \\
                    1 &  \w^{12} & \w^9 & \w^{14} & \w^{3} & \w^{13} & \w^{3} & \w^{6} & \w^{7} & \w^{9} & \w^{11} & \w^{12} & \w^6 & 0 \\
                    1 &  \w^{13} & \w^{11} & \w^{2} & \w^{7} & \w^4 & \w^{10} & \w^{14} & \w & \w^{5} & \w^{8} & \w^{10} & \w^5 & 0 \\
                    1 &  \w^{14} & \w^{13} & \w^{5} & \w^{11} & \w^{10} & \w^{2} & \w^{7} & \w^{10} & \w & \w^{5} & \w^{8} & \w^4 & 0 \\
                    1 &  1 & 1 & \w^{8} & 1 & \w & \w^{9} & 1 & \w^{4} & \w^{12} & \w^{2} & \w^{6} & \w^3 & 0 \\
                    1 &  \w & \w^2 & \w^{11} & \w^{4} & \w^7 & \w & \w^{8} & \w^{13} & \w^{8} & \w^{14} & \w^{4} & \w^2 & 1 \\
                    1 &  \w^{2} & \w^4 & \w^{14} & \w^{8} & \w^{13} & \w^{8} & \w & \w^{7} & \w^{4} & \w^{11} & \w^{2} & \w & 1
                \end{array}
            \right)   
\end{equation}
    \hrulefill
\end{figure*}
        Note also that $$H_{\C_7(\SSS,\vvv,\infty)}\cdot {\bf y}^T=(\w^{11},\w^{11},\w^{10},\w^3,\w^6,\w^6,\w^9)^T\neq {\bf 0},$$ 
        so ${\bf y}$ is not a codeword of $\C_7(\SSS,\vvv,\infty)$.

        \item Modifying Lines 8 to 14 of Algorithm \ref{alg.1} based on Table \ref{tab.ECP}, we have $\ell=3$ and 
        $\uuu=(1, \w^{11}, \w^7, \w^{11}, \w^{14}, \w^7, \w^{11}, \w^{13}, \\ \w^{13}, \w^{13}, \w^{14}, \w^{14}, \w^7)$. 
        Then $(\mathcal{A},\mathcal{B})$ is a $3$-ECP of $\C_7(\SSS,\vvv,\infty)$, 
        where  
        $$
        \mathcal{A}=\{(f(a_1),\ldots,f(a_n),0):~f(x)\in \F_q[x]_4\}
        $$
        \mbox{and}~
        $$
        \mathcal{B}=\GRS_{3}(\SSS,{\bf v}^{-1}\star {\bf u},\infty).
        $$
        Moreover, we deduce \eqref{eq.GA_ex2} and \eqref{eq.GB_ex2}. 

        \begin{figure*}[!t]
        \begin{equation}\label{eq.GA_ex2}
        G_{\mathcal{A}}=\left(
            \begin{array}{cccccccccccccc} 
                1 &  1 & 1 & 1 & 1 & 1 & 1 & 1 & 1 &  1 & 1 & 1 & 1 & 0 \\
                1 & \w & \w^2 & \w^3 & \w^4 & \w^6 & \w^7 & \w^8 & \w^9 & \w^{11} & \w^{12} & \w^{13} & \w^{14} & 0 \\
                1 & \w^2 & \w^4 & \w^6 & \w^8 & \w^{12} & \w^{14} & \w & \w^3 & \w^{7} & \w^{9} & \w^{11} & \w^{13} & 0 \\
                1 & \w^3 & \w^6 & \w^9 & \w^{12} & \w^{3} & \w^{6} & \w^{9} & \w^{12} & \w^{3} & \w^{6} & \w^{9} & \w^{12} & 0 \\
        \end{array}
        \right) 
        \end{equation}   
        \hrulefill
        \end{figure*}

        \begin{figure*}[!t]
        \begin{equation}\label{eq.GB_ex2}
        G_{\mathcal{B}}=\left(
            \begin{array}{cccccccccccccc} 
                1 &  1 & 1 & 1 & 1 & 1 & 1 & 1 & 1 &  1 & 1 & 1 & 1 & 0 \\
                1 & \w^{12} & \w^9 & \w^{14} & \w^3 & \w^{13} & \w^3 & \w^6 & \w^7 & \w^{9} & \w^{11} & \w^{12} & \w^{6} & 0 \\
                1 & \w^{13} & \w^{11} & \w^2 & \w^7 & \w^{4} & \w^{10} & \w^{14} & \w & \w^{5} & \w^{8} & \w^{10} & \w^{5} & 1 \\
        \end{array}
        \right) 
        \end{equation}   
        \hrulefill
      \end{figure*}

       \item With Line 15 of Algorithm \ref{alg.1}, we solve the system of linear equations 
       $$G_{\mathcal{B}}\cdot \diag({\bf y}) \cdot G_{\mathcal{A}}^T\cdot {\bf s}^T={\bf 0}$$
       and get ${\bf s}({\bf y})=\{\mu(1,1,\w^8,\w^4):~ \mu \in \F_{16}\}$.  
       Take ${\bf s}_0=(1,1,\w^8,\w^4)$, then we have  
       \begin{align*}
        {\bf a} & ={\bf s}_0 G_{\mathcal{A}} \\ & =(\w^5,\w^{12},\w^{13},\w^{13},\w,0,\w^5,\w^4,\w^{10},\w^8,\w^{13},\w,\w^5,0)
       \end{align*}
       {\rm and}
       $$Z=z({\bf a})=\{6,14\}.$$

       \item Based on Line 22 of Algorithm \ref{alg.1}, we further set 
       ${\bf x}=(0,0,0,0,0,x_6,0,0,0,0,0,0,0,x_{14})$ and 
       find that the system of linear equations  
       $$H_{\C_7(\SSS,\vvv,\infty)}{\bf x}^T=H_{\C_7(\SSS,\vvv,\infty)}{\bf y}^T$$ 
       has no solution. 
       According to Line 27 of Algorithm \ref{alg.1}, we immediately conclude that 
       the received vector ${\bf y}$ incurs more than $3$ errors. 
    \end{itemize}
\end{example}

\section{Concluding remarks}\label{sec.concluding remarks}

ESGRS codes $\C_k(\SSS,{\bf v},\infty)$ constitute a class of linear codes 
where each code is either a non-GRS MDS code or an NMDS code, 
both of which have significant applications in distributed storage systems and cryptography. 
Although various properties of these codes, including their dual codes, Schur squares, weight distributions, 
and (almost) self-dual properties, were investigated in \cite{LSZ2025,ADV2025},  
their decoding challenges have remained largely unexplored.

This paper aims to bridge this gap, i.e., to answer Problem \ref{prob2}. We fully examined the existence of 
$\ell$-ECPs $(\mathcal{A},\mathcal{B})$ for $\C_k(\SSS,{\bf v},\infty)$ and provided explicit constructions for 
$\mathcal{A}$ and $\mathcal{B}$, either via subcode relationships or through direct constructions, 
as established in Theorems \ref{th.ECP_of_C_k} and \ref{th.ECP_of_C_k222}. 
Note that these ECPs are composed of GRS codes and their variants, and are always defined over the base field. 
Based on these $\ell$-ECPs, we proposed a decoding algorithm for $\C_k(\SSS,{\bf v},\infty)$ that can uniquely 
correct up to $\ell$ errors in $O(n^3)$ time, as detailed in Theorem \ref{th.decoding ECP} and Algorithm \ref{alg.1}. 
As discussed in Subsection \ref{subsec.complexity}, our decoding algorithm is more efficient than some alternative approaches. 

We also highlight several research directions that warrant further exploration, summarized as follows. 

\begin{enumerate}
    \item [1)] In Subsection \ref{subsec.complexity}, we mentioned two potential approaches to further reduce the decoding complexity 
    of Algorithm \ref{alg.1}.  
Therefore, it would be interesting to design $\SSS$ and ${\bf v}$ such that the above two approaches can be reached 
for some special error patterns. 

\item [2)] Motivated by Remark \ref{rem.SGRS}, it is of interest to investigate the conditions under which ESGRS codes 
coincide with SGRS codes, and when they diverge. 
Understanding these conditions is useful both for directly deriving decoding algorithms for SGRS codes 
and for constructing new families of non-GRS MDS codes. 
\end{enumerate}

\vfill

\end{document}